\documentclass[journal,twoside,final]{IEEEtran}
\pdfoutput=1
\usepackage{amsthm,amssymb,graphicx,multirow,amsmath,color,amsfonts}
\usepackage[update,prepend]{epstopdf}
\usepackage[noadjust]{cite}
\usepackage[latin1]{inputenc}

\def\nb0{{\mathbf{0}}}
\def\nb1{{\mathbf{1}}}

\def\ncalT{{\mathcal{T}}}

\def\nrmd{{\rm d}}

\newtheorem{lemma}{Lemma}

\newtheorem{ndef}{Definition}

\newtheorem{theorem}{Theorem}

\newtheorem{cor}{Corollary}
\newtheorem{example}{Example}
\newtheorem{remark}{Remark}

\def\E{\mathbb{E}}

\def\P{\mathbb{P}}
\def\pc{\mathtt{P_c}}

\def\R{\mathbb{R}}

\def\T{\beta}							
\def\sinr{\mathtt{SINR}}			

\def\sir{\mathtt{SIR}}

\def\calA{\mathcal{A}}
\def\calK{\mathcal{K}}
\def\calT{\mathcal{T}}
\def\calB{\mathcal{B}}

\allowdisplaybreaks

\usepackage{url}

\usepackage{flushend}

\begin{document}
\graphicspath{{./Figures/}}

\title{Load-Aware Modeling and Analysis of Heterogeneous Cellular Networks}

\author{Harpreet S. Dhillon,~\IEEEmembership{Student Member, IEEE}, Radha Krishna Ganti,~\IEEEmembership{Member, IEEE}, \\ and Jeffrey G. Andrews,~\IEEEmembership{Fellow, IEEE}\\
\thanks{Manuscript received April 4, 2012; revised October 12, 2012; accepted January 4, 2013. The associate
editor coordinating the review of this paper and approving it for
publication was L. Cai.}
\thanks{This work was supported by NSF grant CIF-1016649. A part of this paper was presented at IEEE Globecom 2012 in Anaheim, CA~\cite{DhiGanC2012}.}
\thanks{
H. S. Dhillon and J. G. Andrews are with the Wireless Networking and Communications Group (WNCG), The University of Texas at Austin (e-mail: dhillon@utexas.edu, jandrews@ece.utexas.edu).}
\thanks{R. K. Ganti is with the Department of EE, Indian Institute of Technology Madras, Chennai, India (e-mail: rganti@ee.iitm.ac.in).}
}

\maketitle




\begin{abstract}
Random spatial models are attractive for modeling heterogeneous cellular networks (HCNs) due to their realism, tractability, and scalability.  A major limitation of such models to date in the context of HCNs is the neglect of network traffic and load: all base stations (BSs) have typically been assumed to always be transmitting.  Small cells in particular will have a lighter load than macrocells, and so their contribution to the network interference may be significantly overstated in a fully loaded model. This paper incorporates a flexible notion of BS load by introducing  a new idea of {\em conditionally thinning} the interference field.  For a $K$-tier HCN where BSs across tiers differ in terms of transmit power, supported data rate, deployment density, and now load, we derive the coverage probability for a typical mobile, which connects to the strongest BS signal. Conditioned on this connection, the interfering BSs of the $i^{th}$ tier are assumed to transmit independently with probability $p_i$, which models the load.  Assuming -- reasonably -- that smaller cells are more lightly loaded than macrocells, the analysis shows that adding such access points to the network always increases the coverage probability.  We also observe that fully loaded models are quite pessimistic in terms of coverage.
\end{abstract}


\section{Introduction}
\PARstart{A}{dvances} in hardware and the increasing popularity of smartphones and tablets have led to a paradigm shift in the way cellular networks are accessed and consequently the way they are deployed. In terms of network access, focus has shifted from voice-oriented applications towards data-hungry applications such as live video streaming and symmetric video calls~\cite{CisM2012}. Macrocell based conventional cellular networks were primarily designed to provide coverage and are clearly not capable of accommodating this huge change in the usage trends~\cite{QuaM2011}.
 One of the most promising ways to handle this data deluge is to increase the density of the BSs, thereby reducing the frequency reuse distance and hence improving network capacity~\cite{AndClaJ2012}. For example, a typical $3$G or $4$G cellular network already has operator-managed picocells deployed at the hot spots and cell edges~\cite{KisGreJ2003}; distributed antennas deployed to eliminate coverage dead-zones~\cite{RohPauC2003}; and low-power user-deployed femtocells~\cite{AndClaJ2012}; along with the existing high-power tower-mounted traditional macrocell BSs that guarantee universal coverage.

\subsection{Related Work and Motivation}
The rapidly increasing heterogeneity and randomness in cellular networks threaten classical models, such as the hexagonal grid~\cite{RapB2002} or Wyner model~\cite{WynJ1994}, with obsolescence.  Clearly, a more sensible way to model a HCN is with a random spatial model, where the BS locations form a realization of some random spatial point process~\cite{BroJ2000,BacZuyBC1997,AndBacJ2011,FleStoM1997}.  Such a model captures the inevitable uncertainty in their locations, and tools from stochastic geometry~\cite{StoKenB1995} and point process theory~\cite{KinB1993} can be deployed to assist in analysis~\cite{BacBlaB2009}.

\pubidadjcol

For example, in an HCN, macrocells would usually follow a somewhat sparse point process and have high transmit power, whereas pico and femtocells are drawn from successively denser (more BSs/area) processes, and have lower transmit power.  In such a network, a mobile user could simply connect to the strongest base station signal, with the rest of the transmitting BSs being interferers.  This model was introduced in \cite{DhiGanC2011,DhiGanJ2012} and extended in \cite{MukJ2012,JoSanJ2012,DhiGanC2011a}, and is surprisingly tractable: under fairly benign assumptions, the coverage probability could be derived in closed-form, which is not possible even for 1-tier networks in the hexagonal grid model.  The model further was shown to generally agree in several important ways with more sophisticated industry (e.g. 3GPP) simulations \cite{DamMonJ2011} and even early field deployments of HCNs \cite{GhoAndJ2012}.

Despite this encouraging progress, these models lack in at least one important aspect, which is their neglect of network traffic and load. Rather, the work to date in this direction has assumed that all the BSs transmit concurrently all the time, which translates to a fully loaded (or full buffer) scenario resulting in pessimistic estimates of coverage and average rate. Although this might be justified for macrocells in peak traffic hours, this is not applicable for smaller cells whose smaller coverage areas will naturally accommodate fewer users, even if considerable biasing towards the small cells is introduced. Therefore, the main goal of this paper is to incorporate a notion of BS load. Those familiar with random spatial models will recognize that a simple independent thinning of the point processes will not capture the load since it may also turn off the serving BS, which is not allowed if the analysis is performed for a typical active user. On the other hand, incorporating more sophisticated queueing models in the present multi cell scenario will render the analysis intractable due to the interference induced coupling in the service rates of various BSs~\cite{RenVecJ2011,BorHegJ2009}. Moreover, this line of thought is not in the scope of the current paper since we do not focus on the flow level performance evaluation. The readers interested in flow level models can refer to~\cite{BorC2003,BonProC2003}. With our main focus on the downlink coverage evaluation, we propose a middle way whereby we conditionally thin the interference field predicated on a connection to a typical active user, and we are able to maintain acceptable tractability with a realistic model of BS loading.

\subsection{Contributions and Outcomes}
The main contributions of this paper are as follows:
\subsubsection{Tractable Load Model for $K$-Tier HCNs} In Section~\ref{sec:sysmod}, we incorporate a notion of BS load in a general $K$-tier random spatial model for HCNs. For an HCN where BSs across tiers differ in terms of their transmit power, supported data rate and deployment density, we assume that a typical mobile connects to the strongest BS in terms of received power and conditioned on this connection, the $i^{th}$ tier interfering BSs transmit independently with a probability $p_i$, which models the load. These BS activity factors $\{p_i\}$ may vary significantly across the tiers due to different coverage areas of each tier. 
\subsubsection{Coverage Probability} We derive exact expressions for the coverage probability of a typical mobile user in both open and closed access HCNs. Since these expressions involve an infinite summation, we also derive a set of upper and lower bounds that can be made arbitrarily tight with a finite number of terms. These bounds also give insights into the number of terms of the infinite summation required to approximate the coverage probability such that the approximation error is within some predefined limit.
\subsubsection{Design Insights} This paper provides some potentially useful design insights for HCNs. First, we study the effect of proposed ``conditional thinning'' on the coverage footprints of various tiers and show that this effect can be understood in two equivalent ways: i) thinning of interference, and ii) biasing of the typical mobile towards its serving BS. While the former is a direct result of thinning, the latter is an indirect consequence of the expansion of the coverage regions in the thinned interference field. 

Second, our analysis sheds light into the effect of adding new tiers to already existing HCNs. In particular, we derive an exact condition under which the addition of a new tier to a general $K$-tier HCN will increase the overall coverage probability. A relevant special case is the addition of small cells to existing macrocell networks, where we show that in the interference limited regime the overall open access coverage probability increases if the load on small cells is smaller than that of macrocells, which is a typical operating scenario because of the smaller loads handled by small cells. This is a strong rebuttal to the viewpoint that unplanned infrastructure might bring down a cellular network due to increased interference.

Third, we show that the coverage probability for a general $K$-tier interference-limited open-access network is invariant to changes in the power and deployment density when all the classes of BSs have same loads and target $\sinr$s. Furthermore, this coverage probability is also the same as that of a single tier network with the same target $\sinr$ and the same BS activity factor.

\section{System Model}
\label{sec:sysmod}
We model a downlink heterogeneous cellular network with $K$ classes (or tiers) of BSs. For notational simplicity, we denote the set $\{1, 2, \ldots K\}$ by $\calK$. BSs of the $i^{th}$ class transmit with power $P_i$, have a target $\sinr$ of $\T_i$ and are assumed to form a realization of an independent homogeneous Poisson Point Process (PPP) $\Phi_i$ with density $\lambda_i$. Such a model seems sensible for user deployed BSs such as femtocells but is dubious for the centrally planned tiers such as macrocells. Nevertheless, the difference is not as large as expected and PPP assumption for macrocells is shown to be about as accurate as the grid model when compared to an actual $4$G network in~\cite{AndBacJ2011}. More recently,~\cite{TayDhiC2012} has validated the PPP assumption for certain cities using tools from spatial statistics. Furthermore, this model is likely even more sensible for $K$-tier HCNs due to the increased uncertainty in the deployment of lower tiers (smaller cells). We will comment more on the accuracy of this assumption in the light of the proposed load model in the Numerical Results Section.

Without loss of generality, we perform analysis on a typical mobile user located at origin, which is made possible by Slivnyak's Theorem~\cite{StoKenB1995}. For cell association, we consider the max-$\sinr$ connectivity model, where a mobile user connects to the BS that provides highest downlink $\sinr$. It should be noted that this model is the same as the max-power connectivity model where a mobile connects to the BS that provides highest downlink power. Since HCNs are typically interference-limited~\cite{BouPanJ2009}, we ignore thermal noise for notational simplicity. However, as would be evident from the analysis, this assumption can be relaxed without much extra work. To model the wireless channel, we consider a standard distance based path loss with exponent $\alpha>2$ along with Rayleigh fading. Hence the received power at a typical mobile from a BS located at point $x \in \Phi_i$ can be expressed as $ P_i h_x \|x\|^{-\alpha}$, where $h_x \sim \exp(1)$ and $ \|x\|^{-\alpha}$ is the distance based path loss. General fading distributions, e.g., log-normal shadowing, can be incorporated using techniques developed in~\cite{BacMuhJ2009} at the cost of tractability. Assuming $\mathcal{Z}_k$ to be the set of $k^{th}$ tier interfering BSs (possibly thinned version of $\Phi_k$), the downlink $\sir$ at the typical mobile user when it connects to the BS located at point $y \in \Phi_i$ is:
\begin{equation}
\sir(y) = \frac{P_i h_{y} \|y\|^{-\alpha}}{\sum_{k=1}^K \sum_{x \in \mathcal{Z}_k} P_k h_{x} \|x\|^{-\alpha} }.
\end{equation}

\subsection{Modeling Base-Station Load}
In this $K$-tier random spatial model, we now incorporate network ``load'' perceived by each BS as the likelihood of its transmission at a randomly chosen time instant. This can also be visualized as the {\em BS activity factor}, formally defined as the fraction of time for which a BS transmits.

\subsubsection*{Relationship of BS activity factor with number of active users}
A BS is inactive in a particular resource block, e.g., time-frequency resource block in LTE~\cite{GhoZhaB2010}, if there is no active user scheduled. This can be due to an over provisioned system or a momentary lull in traffic due to the bursty nature of data access. 
Clearly, this model characterizes  the load on each BS in terms of the total number of active users served by that BS at a random time instant. In the context of Orthogonal Frequency Division Multiple Access (OFDMA) if a particular BS experiences high load, it will utilize more frequency time resources and hence the probability that a user is scheduled in a particular frequency time block increases. Therefore, the load perceived by a BS is directly related to the likelihood that an arbitrary resource block is utilized and hence is related to the BS activity in that particular block.

\subsubsection*{Temporal and spatial correlation in BS activity factors} In general, there is both temporal and spatial correlation in the activity factors of different BSs. Temporal correlation is induced across neighboring BSs by the mobility of users, i.e., if a user is associated to a particular BS, the likelihood of neighboring BSs transmitting at a future time instant is slightly higher. Spatial correlation is induced by interference and traffic/load patterns~\cite{RenVecJ2011,BorHegJ2009}. To understand this, consider two neighboring BSs. When the first BS transmits, it increases net interference experienced by the second BS and hence reduces its data rate. As a result, the second BS now takes longer to transmit same amount of data than it would have taken if the first BS was not transmitting. Therefore, the activity factors of these two BSs are positively correlated. However, modeling the exact nature of these correlations is beyond the scope of the current paper and we assume the BS activity factors to be independent. Although the spatio-temporal correlations haven't yet been modeled for this exact problem, it is worth noting that they have been handled in some related setups, e.g., the effect of spatio-temporal correlations of interference on coverage is discussed in~\cite{GanHaeJ2009,SchBetJ2012}.

\subsection{Proposed Load Model and Mathematical Preliminaries}
We assume that a typical mobile connects to the strongest BS in terms of received power and conditioned on this connection, the interferer belonging to the $i^{th}$ tier transmits independently with a probability $p_i$ and is idle with a probability $1-p_i$. This conditioning makes it harder to analyze this system model since we do not have {\em a priori} knowledge about the serving BS and hence it is not possible to isolate the interference field. To overcome this, we partition each tier $\Phi_m$ independently into two sets of BSs $\Psi_m$ and $\Delta_m$, where $\Psi_m$ and $\Delta_m$ are both independent PPPs with densities $p_m\lambda_m$ and $(1-p_m)\lambda_m$. The set $\Psi_m$ represents the set of active BSs of tier $m$ with the possibility of one of them being a serving BS, and $\Delta_m$ represents the set of idle BSs of tier $m$ with an exception that it could also contain the serving BS since partitioning was done independently. The advantage of this partitioning is that the interferers are confined to the set $\Psi = \bigcup_{m\in \calK} \Psi_m$. For ease of notation, we define the maximum signal strength from a set of nodes $\calA$ as
\begin{equation}
M(\calA) = \sup_{x\in \calA}P_\calA h_x \|x\|^{-\alpha},
\end{equation}
and the total received power at the origin from the set of active BSs as:
\begin{equation}
I= \sum_{i=1}^K \sum_{x\in \Psi_i}P_i h_x\|x\|^{-\alpha},
\end{equation}
which denotes the net interference power if $\Psi$ does not include the serving BS and the interference plus signal power if it includes the serving BS. From the definition of $M(\Psi_i)$ and $I$, it is easy to see that $\nb1 \left(\frac{M(\Psi_i)}{I-M(\Psi)}<\beta_i\right) = 1$ only if no active BS  in the set $\Psi_i$  can connect to the mobile. Similarly, $\nb1 \left(\frac{M(\Delta_i)}{I}<\beta_i\right) = 1$ only if no BS in the set $\Delta_i$ is able to connect to the mobile. The second event is defined to cover the possibility that a serving BS may lie in the set $\Delta_i$. Using these two events, we will now define the coverage probability of a typical mobile at the origin. We note that a mobile will be in outage (not in coverage) if none of the BSs in the whole network provides $\sir$ that is greater than the corresponding target for that tier.
\begin{ndef}[Coverage Probability]
Coverage probability, $\pc$, can be formally defined as:
\begin{equation}
\pc = 1-\E\left[\prod_{i\in \calK}\nb1 \left(\frac{M(\Psi_i)}{I-M(\Psi)}<\beta_i\right)\nb1 \left(\frac{M(\Delta_i)}{I}<\beta_i\right)\right].
\label{eq:coverage}
\end{equation}
\end{ndef}
\noindent For this definition, we implicitly assumed an open access network where a mobile user is allowed to connect to any BS in the network without any restrictions. Another possible access strategy is closed access or closed subscriber group strategy in which a mobile is allowed to connect to only a subset $\calB \subseteq \calK$ of all the tiers. Coverage probability for closed access is also given by \eqref{eq:coverage} with the only difference that the product is over the set $\calB$ instead of $\calK$.

For tractability, we assume that the target $\sir$ thresholds $\beta_i$ are  greater than $0$ dB, i.e., $\beta_i>1$, $\forall\  i$. This is in fact the case for a large fraction of mobile users and only a few edge users might violate this assumption. Moreover, in the Numerical Results Section we show that the results derived under this weaker assumption hold down until around $-2$dB which covers a large fraction of cell edge users as well. This assumption has also been validated earlier for the fully loaded $K$-tier HCN in~\cite{DhiGanJ2012}. The reason why this assumption is helpful is because it  ensures that at most one BS in the active set $\Psi$ meets the target $\sir$ requirements for a typical mobile user. Refer to~\cite{DhiGanJ2012} for a detailed discussion on this assumption and its application in coverage analysis of a fully loaded $K$-tier HCN.

\subsection{Coverage Regions}
\begin{figure*}[t]
\centering
\includegraphics[width=0.32\textwidth]{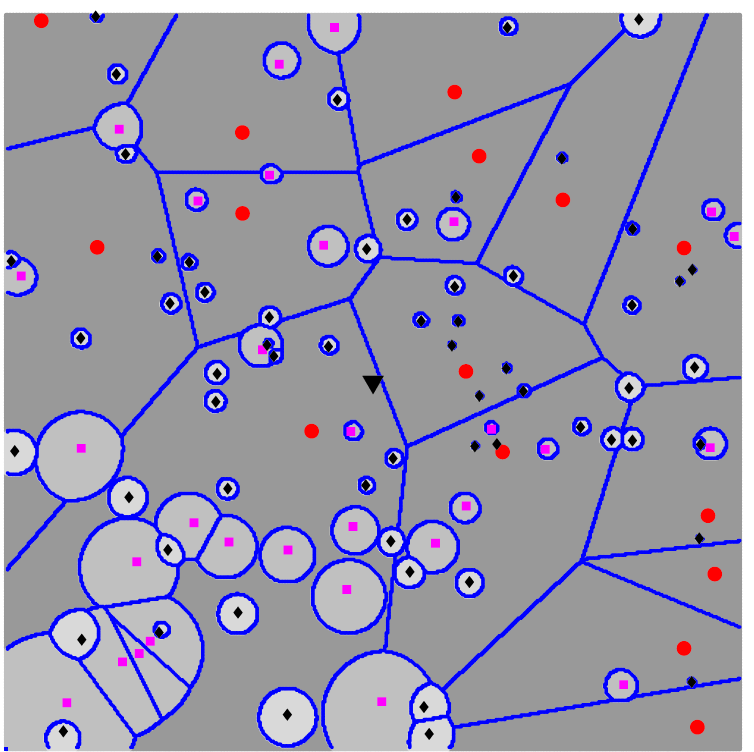}
\includegraphics[width=0.32\textwidth]{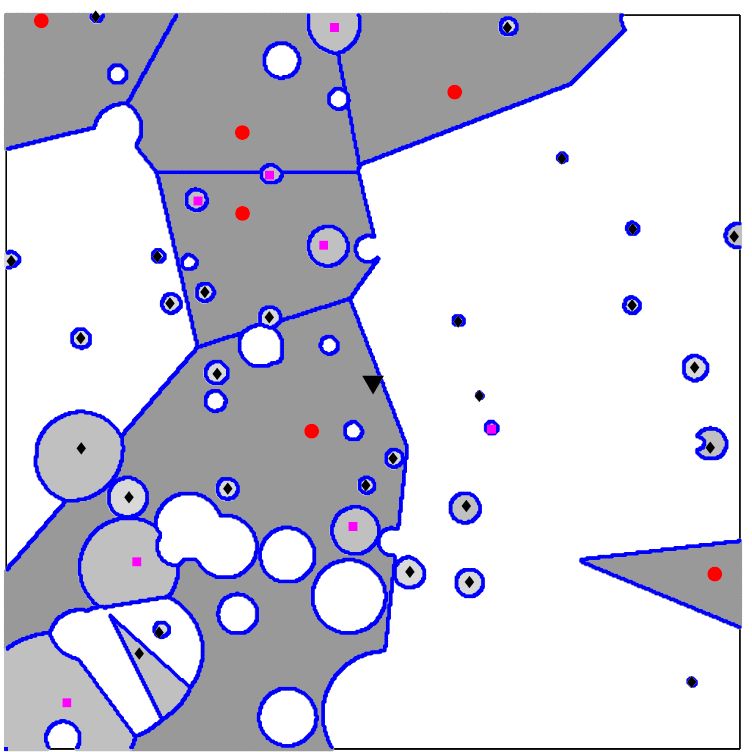}
\includegraphics[width=0.32\textwidth]{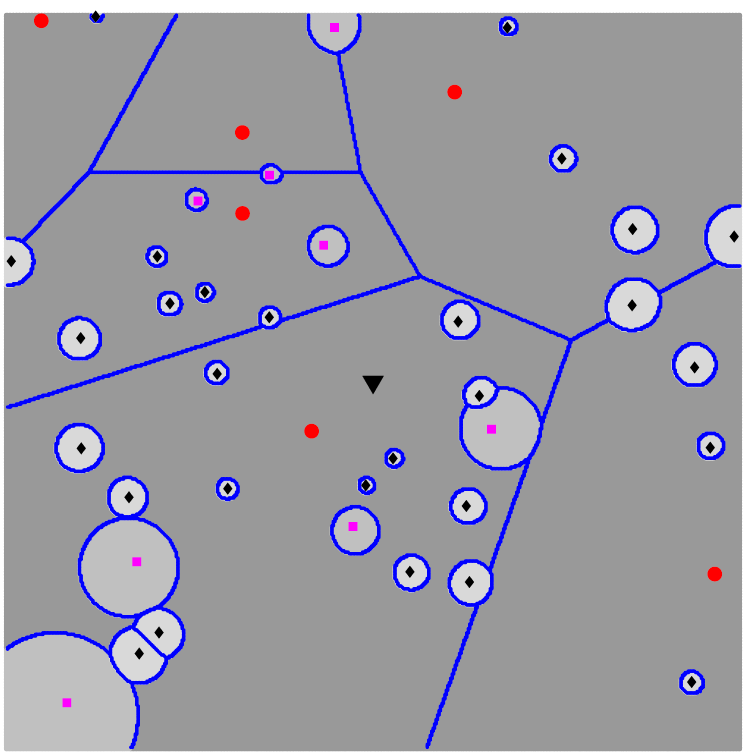}

\caption{Illustration of the proposed load model in a realization of a three-tier network with $\lambda_2 = 2\lambda_1$, $\lambda_3 = 4\lambda_1$, $P_1 = 100P_2$, $P_1 = 1000P_3$, $p_1 = .6$ and $p_2 = p_3 = .4$. The big circles, squares, small diamonds and big triangle, respectively represent macrocells, picocells, femtocells and a typical mobile.}
\label{fig:3tierZout}
\end{figure*}

Before going into the detailed analysis of coverage probability, it will be useful to understand the effect of the proposed load model on the coverage footprints of various BSs. Consider a realization of a three tier HCN in Fig.~\ref{fig:3tierZout}. We first plot the coverage regions assuming a fully loaded network by tessellating the space according to max-$\sir$ connectivity model in the left figure. Clearly, this plot does not resemble a classical Voronoi tessellation due to the differences in the transmit powers of BSs across tiers. Moreover, it should be noted that the ``cell edges'' are not as sharp in reality due to fading and shadowing, which are averaged out for these illustrative plots. The effect of incorporating the proposed load model on coverage footprints can now be understood in two equivalent ways:  i) thinning of the interference field conditional on the connection of a typical mobile to its serving BS, where the original coverage regions corresponding to the inactive BSs are removed to highlight conditional thinning (middle figure), ii) biasing of a typical mobile towards its serving BS relative to the new cell edge defined by the set of active BSs (right figure). While the former is a direct result of conditional thinning, the latter is an indirect consequence of the expansion of coverage regions in the thinned interference field.


\section{Coverage Probability}
This is the main technical section of this paper where we derive the probability that a typical mobile is in coverage under the system model introduced in the last section. We first derive coverage probability for an open access network, from which the results for closed-access immediately follow.

\subsection{Exact Expression for Coverage Probability}

We start by stating the Laplace transform of $I$, i.e., $\mathcal{L}_I(s) = \E \left[\exp(-sI) \right]$, in Lemma~\ref{Lemma2}, which will be useful in the derivation of coverage probability. The proof is given in~\cite{DhiGanJ2012}.
\begin{lemma}
\label{Lemma2}

The Laplace transform of $I$ can be expressed as:
\begin{equation}
\mathcal{L}_{I}(s) = \exp \left( - s^\frac{2}{\alpha} C(\alpha) \sum\limits_{l=1}^K p_l \lambda_l P_l^\frac{2}{\alpha}  \right),
\end{equation}
where $C(\alpha)$ is given by:
\begin{equation}
C(\alpha) = \frac{2 \pi^2 \csc \left(\frac{2 \pi}{\alpha} \right)}{\alpha}.
\end{equation}
\end{lemma}

The following Lemma deals with fractional moments of interference and is the main technical result required to evaluate the coverage probability for this model.
\begin{lemma}
\label{lem:one}
Let $\Psi_i$ denote the set of active transmitters of tier $i$ and $\delta_i=\beta_i/(1+\beta_i)$. Let $I$ denote the total received power from the BSs in the set $\Psi$ and for notational simplicity define $\calT = \nb1 \left(\max\limits_{i\in\calK}\frac{M(\Psi_i)}{\delta_i}<I\right) I^{-2/\alpha}$. Then
\[\E\left[\calT^m\right]= \frac{m!g(m)}{(-A)^m},\]
where
\begin{align}
g(m)=\left(-\frac{A}{\eta}\right)^m
\left\{\frac{1}{\Gamma(1+\frac{2m}{\alpha})}-
\frac{B}{\eta}\frac{\pi\Gamma(1+\frac{2}{\alpha})}{\Gamma(1+\frac{(m+1)2}{\alpha})}\right\},
\label{eq:gm}
\end{align}
and
\begin{align}
A&=\pi\Gamma\left(1+\frac{2}{\alpha}\right)\sum_{l\in \calK}(1-p_l)\lambda_l P_l^\frac{2}{\alpha}\beta_l^{-\frac{2}{\alpha}}, \label{eq:A}\\
B&= \sum_{i\in\calK}\frac{\lambda_ip_iP_i^\frac{2}{\alpha} \beta_i^{-\frac{2}{\alpha}}{}_2F_1(1,\frac{2m}{\alpha},1+\frac{(m+1)2}{\alpha},\frac{1}{1+\beta_i})}{(1+\beta_i)^\frac{2m}{\alpha}} ,   \label{eq:B}\\
\eta&=C(\alpha)\sum_{l=1}^K p_l\lambda_lP_l^{\frac{2}{\alpha}}. \label{eq:eta}
\end{align}
The hypergeometric function is denoted by
${}_2F_1(a,b,c,z)=\frac{\Gamma(c)}{\Gamma(b)\Gamma(c-b)}\int_0^1\frac{t^{b-1}(1-t)^{c-b-1}}{(1-tz)^a}\nrmd t$.
\end{lemma}
\begin{proof}
See Appendix \ref{app:Proof-of-Lemma-main}.
\end{proof}
Using these Lemmas, we now derive the main coverage probability result.
\begin{theorem}[Open Access]
\label{thm:main}
The downlink coverage probability for a typical mobile user in a $K$-tier open access network assuming $\T_i>1,\ \forall\ i,$ is
\begin{align}
\pc=& \frac{\pi}{C(\alpha)} \frac{\sum\limits_{i\in\calK} p_i \lambda_iP_i^{2/\alpha} \T_i^{-2/\alpha}}{\sum_{i=1}^K p_i \lambda_i P_i^{2/\alpha}}-\sum_{m=1}^\infty g(m),
\end{align}
\end{theorem}
\begin{proof}
The coverage probability is given by \eqref{eq:coverage}. Since the point processes $\Delta_i$ and the corresponding fading random variables are independent, conditioning on the common interference, we can move the expectation inside the product. Hence
\begin{equation}
1-\pc = \E\left[\prod_{i=1}^K\nb1 \left(\frac{M(\Psi_i)}{I-M(\Psi)}<\beta_i\right)\E\left[ \nb1 \left(M(\Delta_i)<\beta_iI\right)\right] \right],
\end{equation}
where the inner expectation is with respect to the inactive transmitter sets. We first simplify this inner expectation as follows:
\begin{align}
&\E\left[ \nb1 \left(M(\Delta_i)<\beta_iI\right)\right] \nonumber \\
&= \E \left[\prod_{x \in \Delta_i} \nb1\left(P_i h \|x\|^{-\alpha} < \T_i I \right)\right]\\
&\stackrel{(a)}{=} \E \left[ \prod_{x \in \Delta_i} \left(1 - \exp\left( -\T_i P_i^{-1} I \|x\|^{\alpha} \right) \right)  \right]\\
&\stackrel{(b)}{=} \exp\left( -(1-p_i) \lambda_i \int_{x\in \R^2} \exp\left(-\T_i P_i^{-1} I \|x\|^{\alpha}  \right) \nrmd x\right)\\
&\stackrel{(c)}{=} \exp\left(-(1-p_i) \lambda_i \T_i^{-\frac{2}{\alpha}} I^{-\frac{2}{\alpha}} P_i^{\frac{2}{\alpha}} \pi \Gamma\left(1+ {\frac{2}{\alpha}} \right)\right),
\end{align}
where $(a)$ follows form the fact that fading is Rayleigh distributed, i.e., $h \sim \exp(1)$, $(b)$ follows from the probability generating functional (PGFL) of PPP~\cite{StoKenB1995} and $(c)$ follows from some algebraic manipulations to reduce the integral to a Gamma function. Now recalling the expression of $A$ given by \eqref{eq:A}, we can write:
\begin{equation}
1-\pc=\E \left[ \nb1 \left(\max_{i\in \calK}\frac{M(\Psi_i)}{\delta_i}<I\right)\exp(-A I^{-2/\alpha})\right].
\end{equation}
Using the Taylor series expansion of $\exp(-x)$,
exchanging the infinite summation and expectation\footnote{The average of the series is absolutely convergent. },
\[1-\pc=\sum_{m=0}^\infty\frac{(-A)^m}{m!}\E\left[\nb1 \left(\max_{i\in\calK}\frac{M(\Psi_i)}{\delta_i}<I\right) I^{-2m/\alpha}\right].\]
The summation can be split as:
\begin{equation}
1-\pc=\P\left(\max_{i\in \calK}\frac{M(\Psi_i)}{\delta_i}<I\right)+ \sum_{m=1}^\infty\frac{(-A)^m}{m!}\E\left[\calT^m\right] .
\label{eq:786}
\end{equation}
The term $1-\P\left(\max_i\frac{M(\Psi_i)}{\delta_i}<I\right)$ is the coverage probability in a fully loaded heterogeneous network where the $m$-th tier density is $p_m \lambda_m$. This is derived in~\cite{DhiGanJ2012} and is given by:
\begin{equation}
1-\P\left(\max_i\frac{M(\Psi_i)}{\delta_i}<I\right)=\frac{\pi}{C(\alpha)} \frac{\sum\limits_{i=1}^K p_i \lambda_i P_i^{2/\alpha} \T_i^{-2/\alpha}}{\sum_{i=1}^K p_i \lambda_i P_i^{2/\alpha}}.
\end{equation}
Using Lemma \ref{lem:one} to evaluate $\E\left[\calT^m\right]$, we obtain the result.
\end{proof}
We note that the expression of coverage probability involves infinite summation over the sequence $g(m)$. Therefore, we first show that the infinite summation converges by showing that $|g(m)|\rightarrow 0$ as $m\rightarrow \infty$. Observe that:
\begin{align}
|g(m)| &\leq  \left(\frac{A}{\eta}\right)^m \frac{1}{\Gamma(1+\frac{2m}{\alpha})} \leq \frac{(A/\eta)^m}{\lfloor 1 + \frac{2m}{\alpha} \rfloor !} \nonumber \\
&= \frac{(A/\eta)^m}{\lceil \frac{2m}{\alpha} \rceil !} = \frac{\left[ \left(A/\eta \right)^\frac{m}{\lceil \frac{2m}{\alpha} \rceil} \right]^{\lceil \frac{2m}{\alpha} \rceil}  }{\lceil \frac{2m}{\alpha} \rceil!} \rightarrow 0,
\end{align}
where the limiting argument follows from the fact that the sequence of the form $x^n/n! \rightarrow 0$. In addition to proving that the series converges, this upper bound on $|g(m)|$ also sheds light on the behavior of the sequence $g(m)$. If $A/\eta < 1$, the bound decreases monotonically with $m$ and hence it is enough to consider only  a few significant terms to closely approximate the infinite sum. However, if $A/\eta > 1$, especially if $A/\eta \gg 1$, the upper bound first increases until $\lceil\frac{2m}{\alpha} \rceil \leq (A/\eta)^\frac{m}{\lceil \frac{2m}{\alpha}\rceil}$ and decreases thereafter. Therefore, the number of significant terms of $g(m)$ required to approximate the infinite sum would be higher. It can be easily shown that $A/\eta < 1$ for all choices of system parameters when the activity factor of each tier satisfies the following condition:
\begin{equation}
p_l > \frac{1}{1 + C(\alpha)\T_l^{2/\alpha}\left[\pi \Gamma(1+2/\alpha) \right]^{-1} }.
\end{equation}
For $\T_l = 1$ and $\alpha=4$, this value of $p_l$ comes out to be $\approx 0.36$. Therefore, the infinite sum can be tightly approximated by the first few significant terms of $g(m)$ in most operating scenarios. We will comment more on the convergence of $g(m)$ and the number of terms required to tightly approximate the coverage probability later in this section and in the Numerical Results Section. We now provide the exact expression for the coverage probability in a closed access network in the following Theorem. We recall that that coverage probability in closed-access is given by~\eqref{eq:coverage} with the only change that the product is over $\calB$ instead of $\calK$. By definition, coverage probability in closed access is less than that of open access. Using this definition, the proof proceeds exactly same as that of Theorem~\ref{thm:main}, and hence is not provided.

\begin{theorem}[Closed Access]
\label{thm:main_closed}
The downlink coverage probability of a typical mobile in a $K$-tier closed access network where a mobile is allowed to connect to $\calB\subseteq\calK$ tiers assuming $\T_i>1,\ \forall\ i,$ is
\begin{align}
\pc=& \frac{\pi}{C(\alpha)} \frac{\sum\limits_{i\in \calB} p_i \lambda_iP_i^{2/\alpha} \T_i^{-2/\alpha}}{\sum_{i=1}^K p_i \lambda_i P_i^{2/\alpha}}-\sum_{m=1}^\infty g_c(m),
\end{align}
where $g_c(m)$ and the corresponding expression for $A$ are given by \eqref{eq:gm} and \eqref{eq:A}, respectively, with the only difference that the summations defined over set $\calK$ are now over set $\calB$.
\end{theorem}

We conclude this discussion with a note that the proof technique introduced in this section is quite general and can be used to study variants of the load model introduced in the last section. For example, if the network is modeled such that it has a predefined set of BSs that are active and a typical mobile is allowed to connect only to the inactive set, it is easy to observe that the coverage probability under open access assumption is given by $\pc = 1-\E\left[\prod_{i=1}^K\nb1 \left(\frac{M(\Delta_i)}{I}<\beta_i\right)\right]$. From the proof of Theorem \ref{thm:main}, this corresponds to $1-\E[\exp(-AI^{-2/\alpha})]$ and can easily be evaluated following the proof technique of Theorem \ref{thm:main}. The same argument can be extended to the closed access case as well.

\subsection{Special Cases of Interest}
We now use the results derived in this section to study some special cases and compare the system performance with already known results for fully loaded system. First, we note that for a fully loaded system, the value of $A=0$ and hence $g(m)=g_c(m)=0,\ \forall\ m$. Therefore, the coverage probability in this case can be expressed as the following Corollary of Theorems~\ref{thm:main} and~\ref{thm:main_closed}.

\begin{cor}[Fully Loaded]
\label{cor:full-load}
For a fully loaded system, i.e., $p_i=1\ \forall\ i$, the coverage probability in open access is given by:
\begin{equation}
\pc = \frac{\pi}{C(\alpha)} \frac{\sum\limits_{i\in \calK}  \lambda_i P_i^{2/\alpha} \T_i^{-2/\alpha}}{\sum_{i=1}^K \lambda_i P_i^{2/\alpha}},
\label{eq:pc_fl_open}
\end{equation}
which is the same as Corollary $1$ in~\cite{DhiGanJ2012}. The coverage probability in closed access is also given by \eqref{eq:pc_fl_open} with the only difference that the summation over the set $\calK$ is now over set $\calB$.
\end{cor}

For a single tier open access network, the coverage probability derived in Theorem~\ref{thm:main} can be simplified and is expressed as the following Corollary.
\begin{cor}[Single Tier]
\label{cor:Pc1tier}
The coverage probability for the single tier open access network with BS activity factor $p$ is
\begin{align}
\pc=& \frac{\pi \T^{-2/\alpha}}{C(\alpha)} -\sum_{m=1}^\infty g(m),
\label{eq:Pc_sameT}
\end{align}
where the terms $\frac{A}{\eta}$ and $\frac{B}{\eta}$ appearing in the expression of $g(m)$ given by \eqref{eq:gm},
\begin{align}
\frac{A}{\eta} &= \frac{\pi\Gamma(1+\frac{2}{\alpha})(1-p)}{C(\alpha)p \T^\frac{2}{\alpha}} \label{eq:A_eta_simple} \\
\frac{B}{\eta} &= \frac{{}_2F_1(1,\frac{2m}{\alpha},1+\frac{(m+1)2}{\alpha},\frac{1}{1+\T})}{C(\alpha)\T^\frac{2}{\alpha}(1+\beta)^\frac{2m}{\alpha}}. \label{eq:B_eta_simple}
\end{align}
\end{cor}
\begin{remark} [Scale invariance of a single tier network]
From Corollary~\ref{cor:Pc1tier}, we note that for any BS activity factor $p$, the coverage probability in a single tier open access network is independent of the BS density $\lambda$ and transmit power $P$. This is henceforth referred to as ``scale-invariance'' of cellular networks to changes in the BS density and their transmit powers.
\label{rem:1tier}
\end{remark}
Remark~\ref{rem:1tier} is a generalization of a similar result derived for fully loaded networks in~\cite{DhiGanJ2012}, which can easily be seen from Corollary~\ref{cor:full-load}. In addition to single tier networks, it was also observed in~\cite{DhiGanJ2012} that the general fully loaded open access multi tier networks also exhibit scale invariance if the target $\sir$s for all the tiers are the same. This can also be easily deduced from Corollary~\ref{cor:full-load}. Motivated by this observation, we study the coverage probability for our proposed load model in open-access multi tier networks under the assumption that the target $\sir$ is the same for all tiers in the next Corollary.

\begin{cor}[Coverage Probability: $K$-Tier with same $\T$] \label{cor:PcKtier}
The coverage probability for a $K$-tier open access network under the proposed load model assuming target $\sir$s to be the same ($=\T$) for all the tiers is given by \eqref{eq:Pc_sameT}, with the difference that the term $\frac{A}{\eta}$ appearing in the expression of $g(m)$ given by \eqref{eq:gm} is:
\begin{align}
\frac{A}{\eta} = \frac{\pi\Gamma(1+\frac{2}{\alpha})}{C(\alpha)\T^{2/\alpha} } \frac{\sum_{l=1}^K (1-p_l) \lambda_l P_l^{2/\alpha}}{\sum_{l=1}^K p_l \lambda_l P_l^{2/\alpha}},
\end{align}
and $\frac{B}{\eta}$ appearing in \eqref{eq:gm} is given by \eqref{eq:B_eta_simple}.
\end{cor}
\begin{remark}[Scale invariance of $K$-tier HCNs with same $\T$] \label{rem:Ktier}
From Corollary~\ref{cor:PcKtier}, we note that the coverage probability for $K$-tier HCNs is not scale invariant in general, even when target $\sir$s of all the tiers are the same. However, the invariance property does hold when the BS activity factors of all the tiers are the same. Interestingly, the coverage probability in this case is same as that of a single tier network given by Corollary~\ref{cor:Pc1tier}.
\end{remark}
To understand this remark, we consider the following simple example.
\begin{example}[Scale invariance in a 2-tier HCN] \label{eg:2tier}
Consider a two tier network with BS activity factors $p_1$ and $p_2$. If $p_1<p_2$, increasing the density of the first tier leads to a higher increase in the intended power due to the higher likelihood of having a closer tier-$1$ BS as the serving BS but a relatively smaller increase in the interference power. The coverage probability in this case is expected to increase. On the other hand, if $p_1 > p_2$, increasing the density of tier-$1$ BSs leads to higher increase in the interference power as compared to the intended power, leading to a decrease in the coverage probability. The two effects cancel each other when the activity factors of the two tiers are the same. 
\end{example}
We now extend this result and derive exact condition under which the addition of $(K+1)^{th}$ tier won't affect (or will improve) the coverage of the existing $K$-tier network. We again assume same target $\sir$ for all the tiers. The result is given in the following Corollary.
\begin{cor}[Same $\T$: Effect of adding $(K+1)^{th}$ tier] The overall coverage probability increases with the addition of the $(K+1)^{th}$ tier if the load on the new tier satisfies:
\begin{equation}
p_{K+1} < \sum_{l=1}^K p_l \frac{\lambda_l P_l^{2/\alpha}}{\sum_{i=1}^K \lambda_i P_i^{2/\alpha}},
\label{eq:genKadd}
\end{equation}
decreases if the inequality is reversed and remains the same if \eqref{eq:genKadd} holds with equality.
\end{cor}
\begin{IEEEproof}
From Corollary~\ref{cor:PcKtier} we note that the only term in the coverage probability expression that will change with the addition of a new tier is $A/\eta$. It can be expressed as:
\begin{align}
\frac{A}{\eta} = \frac{\pi\Gamma(1+\frac{2}{\alpha})}{C(\alpha)\T^{2/\alpha} } \left(  \left[ \sum_{l=1}^K p_l \frac{\lambda_l P_l^{2/\alpha}}{\sum_{i=1}^K \lambda_i P_i^{2/\alpha}} \right]^{-1}  - 1  \right).
\end{align}
Defining effective load on a $K$-tier network as:
\begin{align}
p_{\rm eff}^{(K)} = \sum_{l=1}^K p_l \frac{\lambda_l P_l^{2/\alpha}}{\sum_{i=1}^K \lambda_i P_i^{2/\alpha}},
\end{align}
$A/\eta$ can be expressed as:
\begin{align}
\frac{A}{\eta} = \frac{\pi\Gamma(1+\frac{2}{\alpha})}{C(\alpha)\T^{2/\alpha} } \left(  \frac{1-p_{\rm eff}^{(K)}}{p_{\rm eff}^{(K)}}  \right),
\end{align}
which is the same as~\eqref{eq:A_eta_simple} for the single tier coverage result derived in Corollary~\ref{cor:Pc1tier}. From this equivalence, it follows that the  coverage probability is a decreasing function of $p_{\rm eff}$. Therefore, if the addition of the new tier leads to lower effective load on the network, the coverage will increase. This can be shown to be the case when \eqref{eq:genKadd} holds as follows:
\begin{align}
p_{\rm eff}^{(K+1)} &= \sum_{l=1}^{K+1} p_l \frac{\lambda_l P_l^{2/\alpha}}{\sum_{i=1}^{K+1} \lambda_i P_i^{2/\alpha}} \\
&\leq  \frac{ \sum_{l=1}^K p_l\lambda_l P_l^{\frac{2}{\alpha}}}{\sum_{i=1}^{K+1} \lambda_i P_i^{\frac{2}{\alpha}}} + \frac{ \lambda_{K+1} P_{K+1}^{\frac{2}{\alpha}}}{\sum_{i=1}^{K+1} \lambda_i P_i^{\frac{2}{\alpha}}}   \frac{ \sum_{l=1}^K p_l\lambda_l P_l^{\frac{2}{\alpha}}}{\sum_{i=1}^{K} \lambda_i P_i^{\frac{2}{\alpha}}}\\
& = p_{\rm eff}^{(K)}.
\end{align}
The other two results follow using the same argument.
\end{IEEEproof}
\subsection{Bounds on the Coverage Probability}
Evaluation of the exact expression of the coverage probability requires an infinite summation. Although we have argued that the summation can be tightly approximated by considering only a first few terms, we haven't yet provided a formal method to determine the exact number of terms required such that the approximation error is within predefined limit, say $\epsilon$. Interestingly, this can be achieved as a by-product of the set of bounds we derive in this section that can be made arbitrarily tight. The idea is to use the following identity of $\exp(-x)$.
\begin{lemma}
For $x\geq0$ and $m>0$,
\begin{equation}
\sum_{i=0}^{2m-1} \frac{(-x)^i}{i!} \leq \exp(-x) \label{eq:boundx}
\leq \sum_{i=0}^{2m} \frac{(-x)^i}{i!}.
\end{equation}
\end{lemma}
\begin{IEEEproof}
The proof follows from induction. Since $\exp(-0)=1$ and $\sum_{i=0}^{2m-1} \frac{(-0)^{i}}{i!}=1$,  it suffices  to prove that $\frac{d}{d x}\sum_{i=0}^{2m-1} \frac{(-x)^{i}}{i!}-\exp(-x) <0$, which follows from the upper bound when $m=m-1$.
\end{IEEEproof}
Using this identity in the proof of Theorem~\ref{thm:main} results in the following bounds.
\begin{lemma}[Bounds on Coverage Probability]
For $m>0$, the coverage probability for the proposed load model can be bounded as
\begin{equation}
- \sum_{i=1}^{2m}g(i)\leq \pc- \frac{\pi}{C(\alpha)} \frac{\sum\limits_{i=1}^K p_i \lambda_i P_i^{2/\alpha}\T_i^{-2/\alpha}}{\sum_{i=1}^K p_i \lambda_i P_i^{2/\alpha} }\leq - \sum_{i=1}^{2m-1}g(i)
\end{equation}
\end{lemma}
Clearly, these bounds can be made arbitrarily tight by increasing the value of $m$. Interestingly, these bounds are closely related to the exact expression of coverage probability derived in Theorem~\ref{thm:main}. In particular, the upper and lower bounds are derived by truncating the infinite sum over $g(m)$ at odd and even number of terms, respectively. Therefore, these bounds provide a direct way to find the number of terms of $g(m)$ required to ensure an approximation error within a predefined limit $\epsilon$, which is equal to $M_{\epsilon}$, where $M_\epsilon = \min\limits_{m} |g(m)| < \epsilon$. We will use this observation in the study of the convergence of infinite sum over $g(m)$ in the Numerical Results Section.

We conclude this section by noting that some terms of the sequence $g(m)$ can be expressed in closed form, leading to closed form bounds for the special case when $\alpha=4$ and $m=2$. The bounds in this case depend only on the first two terms of $g(m)$ that can be expressed as:
\begin{align}
g(1)&=\frac{-A}{\eta}\left\{\frac{2}{\sqrt{\pi} }-\frac{4\sqrt{\pi}}{\eta}\sum_{i=1}^K\frac{\lambda_i p_iP_i^{1/2}\beta_i^{-1/2}}{1+\sqrt{1+\beta_i}}  \right\}\\
g(2)&=\left(\frac{A}{\eta}\right)^2 \nonumber \\ &\left\{1-\frac{2\pi}{\eta}\sum_{i=1}^K\lambda_ip_iP_i^\frac{1}{2}(\beta_i^{-\frac{1}{2}} -\csc^{-1}(\sqrt{1+\beta_i}))\right\}.
\end{align}

\section{Numerical Results}
\begin{figure}[t!]
\centering
\includegraphics[width=\columnwidth]{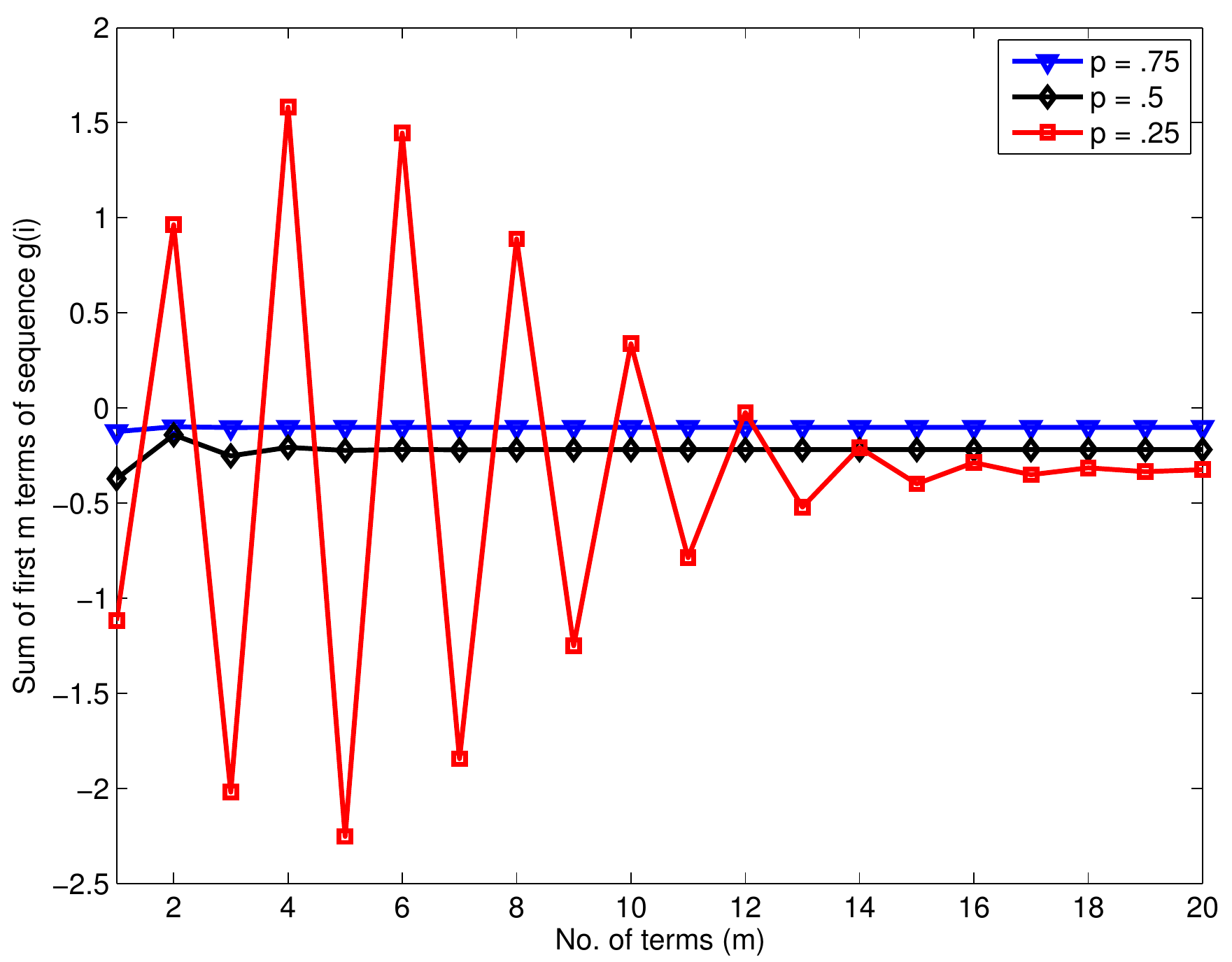}
\caption{Plot showing the convergence of the series $\sum_{i}g(i)$ for various BS activity factors in a single tier network with $\T = 1$.}
\label{fig:Conv_gm}
\end{figure}

\begin{figure}[t!]
\centering
\includegraphics[width=\columnwidth]{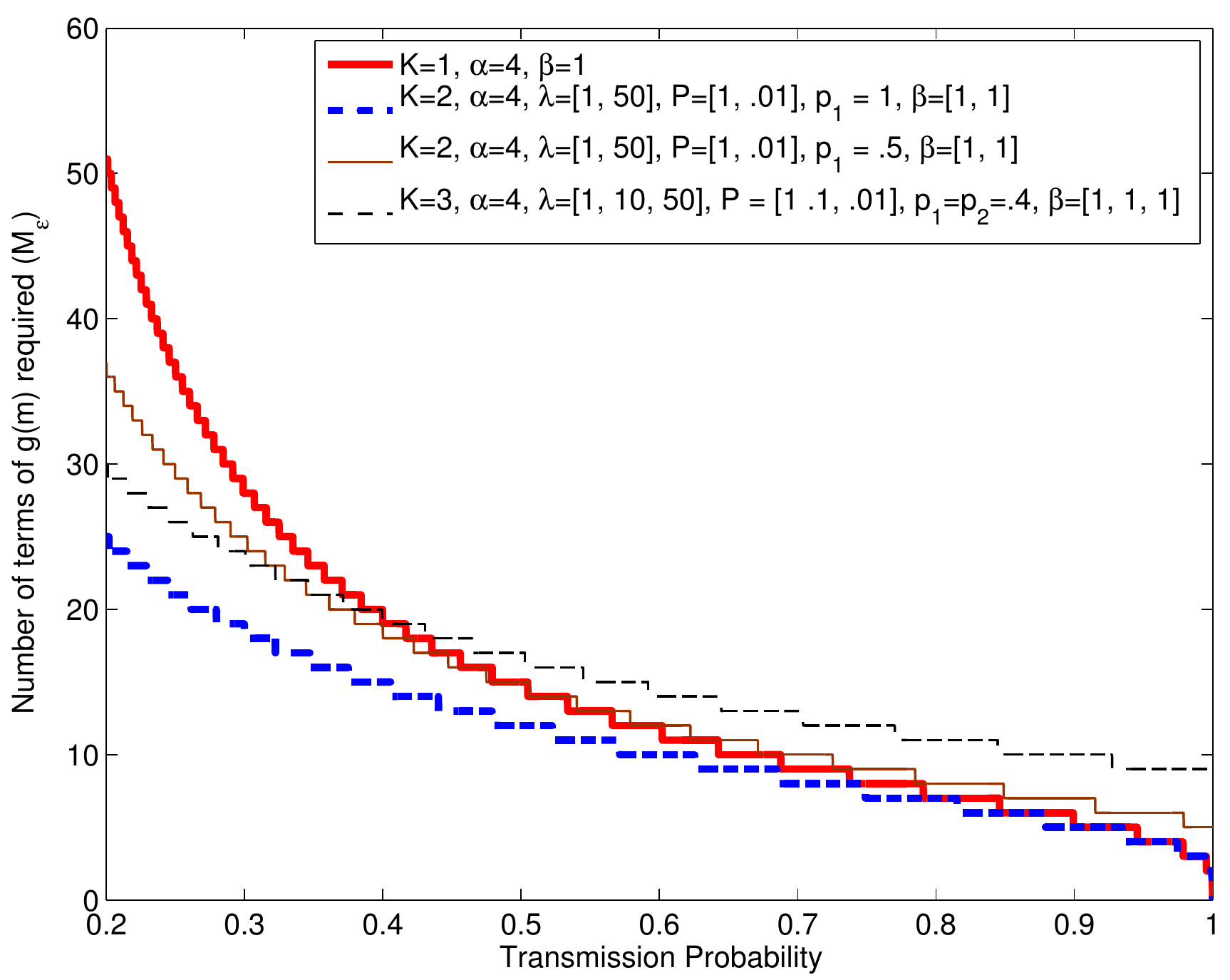}
\caption{Number of terms of the sequence $g(m)$ required as the function of the transmission probability of the lowest tier for $\epsilon = 10^{-8}$.}
\label{fig:NTerms_gm}
\end{figure}

\begin{figure}[t!]
\centering
\includegraphics[width=\columnwidth]{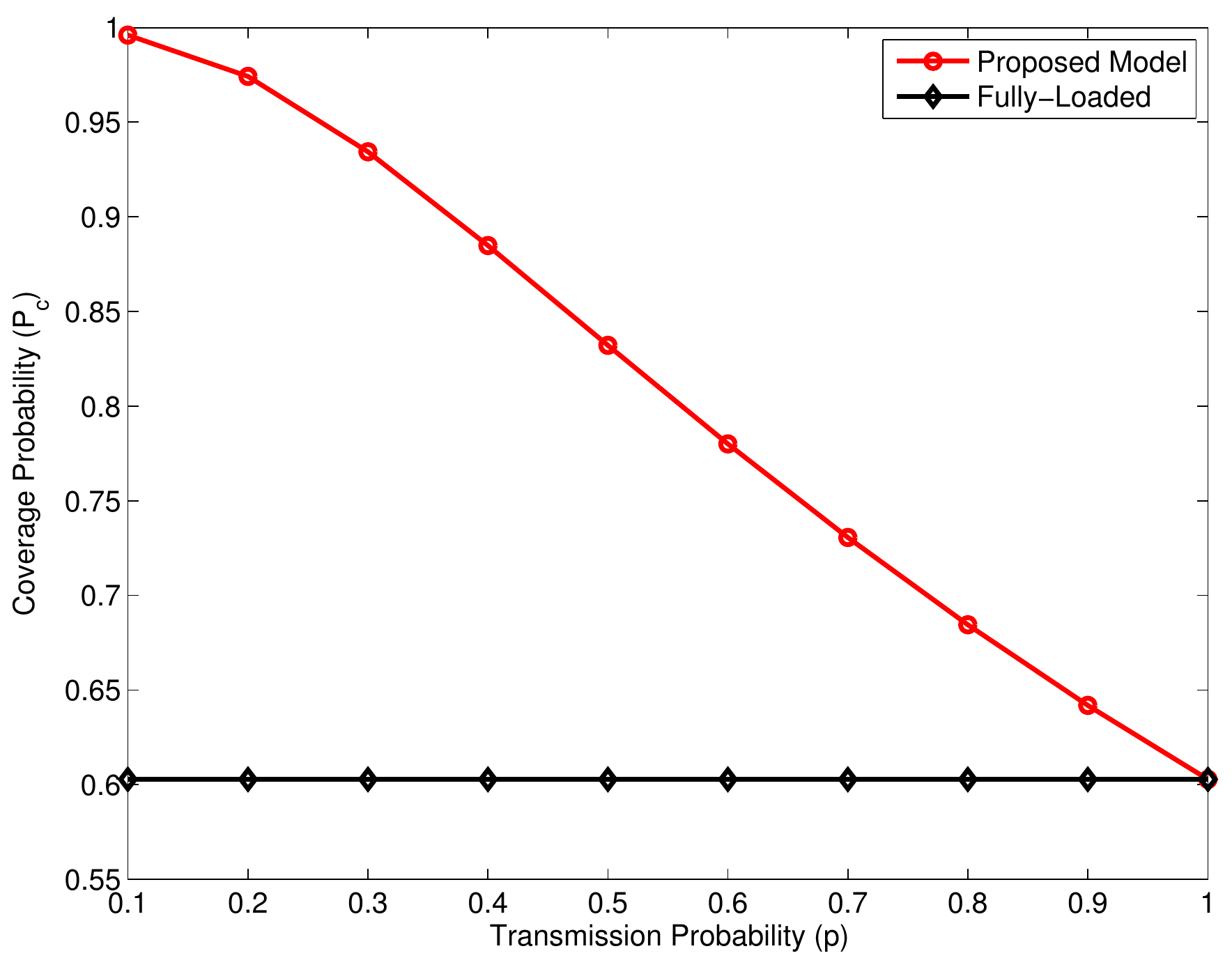}
\caption{Coverage probability as a function of transmission probability is a single tier network ($\T = 1$ and $\alpha = 3.8$).}
\label{fig:FLComp_1tier}
\end{figure}

Since most of the analytical results derived in this paper are fairly self-explanatory and do not require separate numerical treatment, we will provide only those results which help in validating key modeling assumptions or help better visualize certain important trends.
\subsection{Convergence of Infinite Sum}
We study the convergence of the infinite sum appearing in the coverage probability expression in Figs.~\ref{fig:Conv_gm} and~\ref{fig:NTerms_gm}. Fig.~\ref{fig:Conv_gm} plots the truncated series $\sum_{i=1}^m g(i)$ as the function of $m$ for a single tier network and hence gives insights about the number of terms required until the series converges. To understand the trends, recall that the ratio $A/\eta$ decreases monotonically with the activity factor $p$. Therefore, the number of terms required for the series to converge are higher when the BS activity factor is lower. Moreover, for the case when $A/\eta>1$, i.e., $p=.25$, the series first increases until a certain point and then decreases and finally converges to its limiting value. This trend has been discussed in detail earlier in the paper when we proved the convergence of the infinite sum. To provide an idea of the number of terms required such that the approximation is within $\epsilon$ of the exact value, we plot the number of terms $M_\epsilon$ for various scenarios in Fig.~\ref{fig:NTerms_gm}. We again note that the number of terms required are reasonably small unless the transmission probability of some tier is extremely small.

\subsection{System Model Validation}
\begin{figure}[t!]
\centering
\includegraphics[width=\columnwidth]{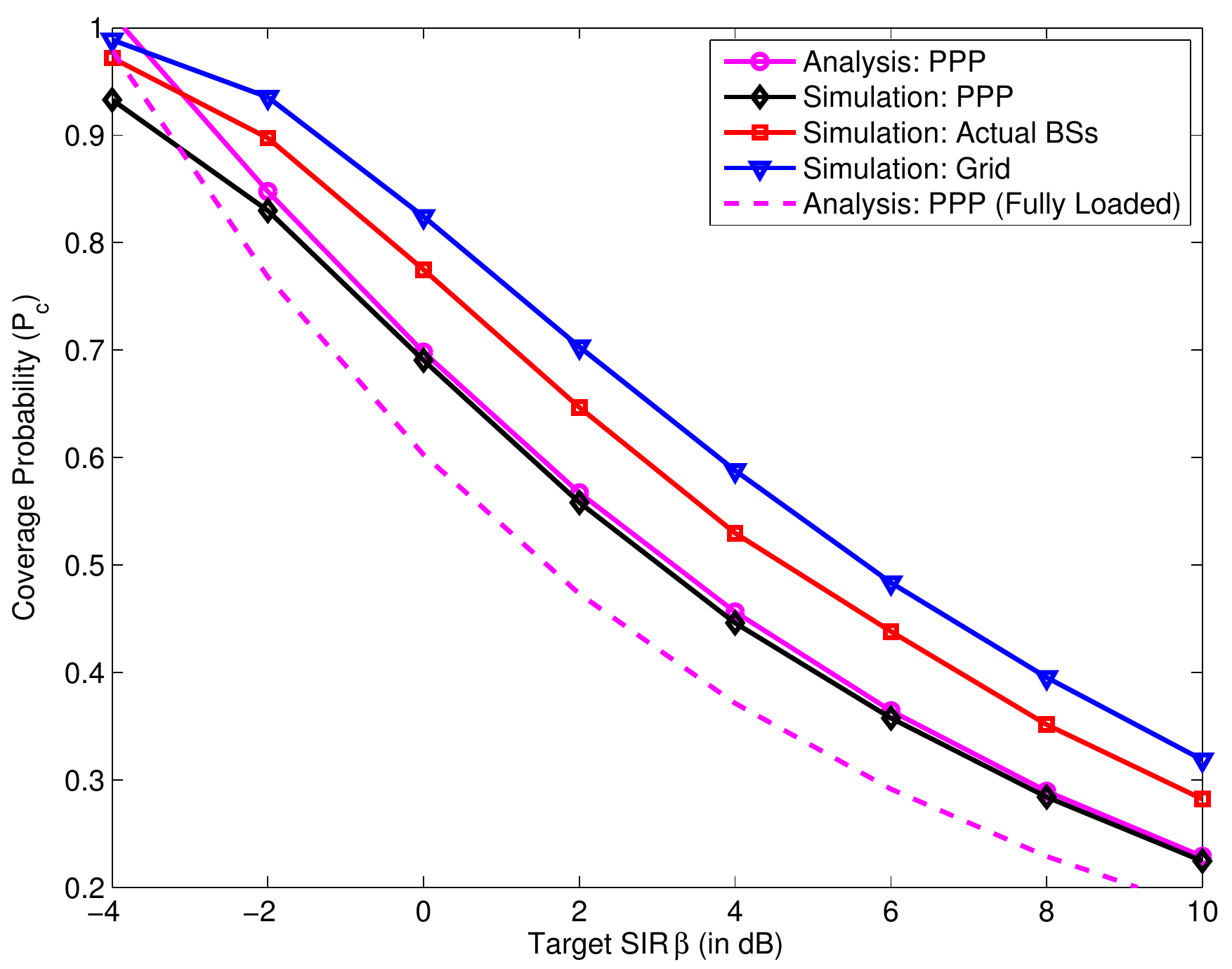}
\caption{Comparison of the coverage probability of the PPP and grid models with the actual BS locations of macrocells. The second tier is modeled as PPP in all three cases ($K=2$, $\lambda_2 = 2\lambda_1$, $P = [1, 0.01]$, $L=40\times 40$ Km$^2$, $\alpha = 3.8$, $p = [0.8, 0.6]$, $\beta_1 = \beta_2 = \beta$).}
\label{fig:PPP_Grid_Actual}
\end{figure}

\subsubsection{Comparison with the fully loaded system}
After gaining insights into the behavior of the coverage probability expression, we now use it to highlight the importance of the proposed model by comparing the coverage results of a single tier network with those of a fully loaded system in Fig.~\ref{fig:FLComp_1tier}. Although a huge difference in the coverage guarantees was expected for very low BS activity factors, it is indeed interesting that the coverage estimates assuming full load are quite pessimistic even for reasonably high load scenarios, such as $p=.7-.8$.
\subsubsection{Comparison with an actual $4$G deployment} After highlighting the importance of the proposed load model, we now validate the PPP model used for the BS locations in this paper. While this model seems sensible for the small cells, especially the ones driven by unplanned user deployments, such as femtocells, it is dubious for the centrally planned tiers such as macrocells. Therefore, with a special focus on the macrocells, we consider three location models for a two tier HCN: i) macrocell locations modeled as a realization of a PPP, ii) macrocell locations modeled as a hexagonal grid, iii) macrocell locations drawn from an actual $4$G deployment over $40 \times 40$ Km$^2$ area~\cite{AndBacJ2011,DhiGanJ2012}. The second tier is modeled as a PPP in all three cases. The numerically evaluated coverage probability results for all these models along with the analytical results of the proposed load model and the fully-loaded PPP model are presented in Fig.~\ref{fig:PPP_Grid_Actual}. We first note that the proposed PPP model is about as accurate as the grid model when compared to the actual $4$G deployment, with the grid model providing an upper bound and the PPP model providing a lower bound to the actual coverage probability. This is consistent with the conclusions of~\cite{AndBacJ2011,DhiGanJ2012}, which focus on fully-loaded cellular models in single tier and multi tier cellular networks, respectively. Second, we note that the analytical results derived for the proposed load model are accurate down to about $-2$dB even though they were derived under the assumption that the target-$\sir$ is greater than $0$dB for all the tiers. Since this covers most of the cell edge users as well, the proposed analytical results are reasonably accurate in the operational regime of the current cellular networks. Third, we note that the fully-loaded model provides a very loose lower bound to the actual coverage probability, thereby highlighting again the importance of the proposed load model.

\begin{figure}[t!]
\centering
\includegraphics[width=\columnwidth]{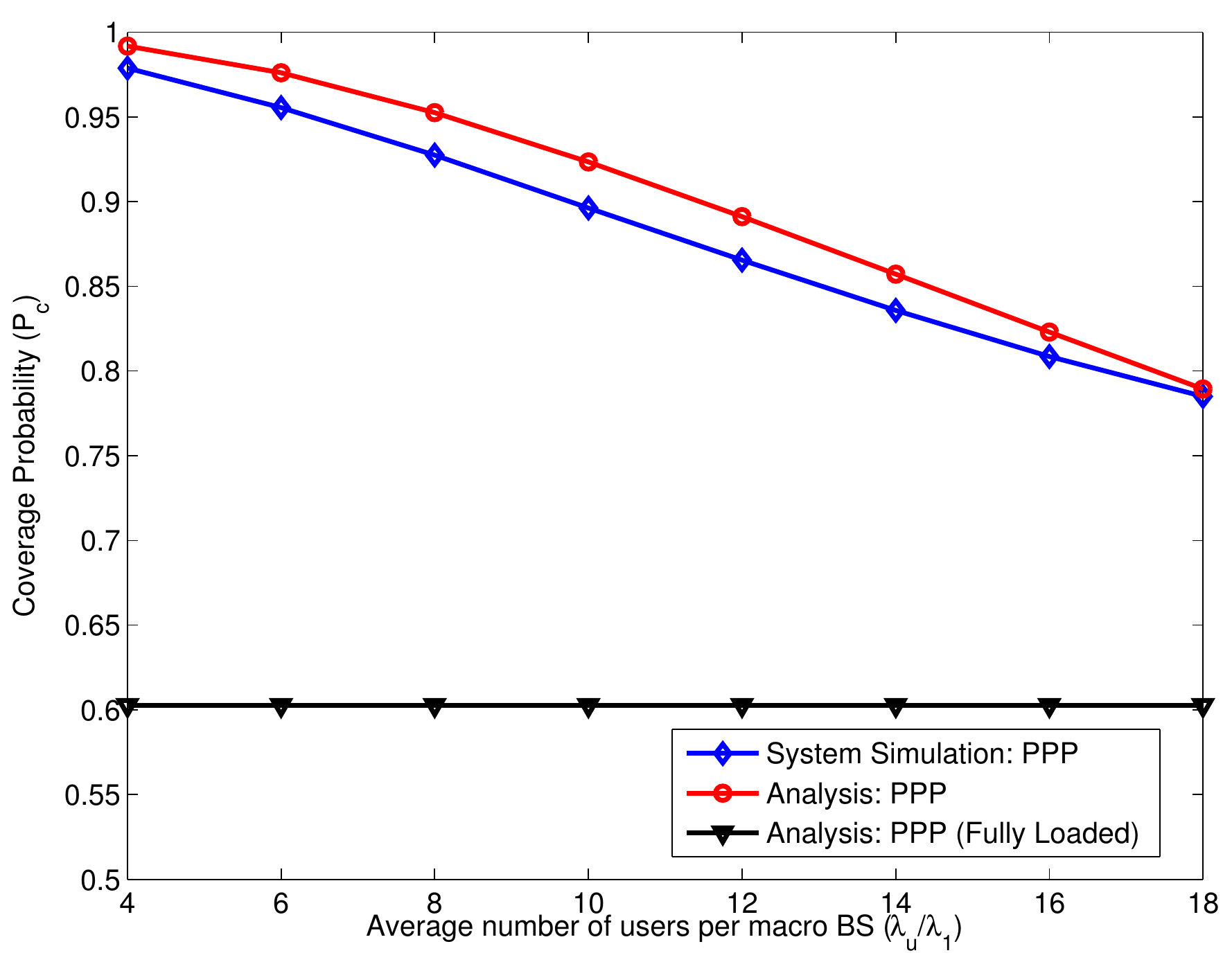}
\caption{Comparison of the derived theoretical results with detailed system simulation accounting for actual load factors resulting from actual coverage regions ($K=2$, $P = [1,0.1]$, $\beta_1 = \beta_2 = 0$dB, $\lambda_1 = \lambda_2$, $M_{RB} = 20$, $\alpha = 3.8$).}
\label{fig:SystemSim}
\end{figure}

\subsubsection{Comparison with a detailed simulation} The following two assumptions were made to facilitate analysis: i) the activity factors are the same for all the BSs of a particular tier, and ii) the activity of each BS is independent of the other BSs.
We now validate these assumptions by comparing the analytical results with a detailed system simulation. For this comparison, we consider a simulation setup consisting of a two tier HCN, with the BSs of each tier modeled by independent PPPs. The user locations are also modeled by an independent PPP with density $\lambda_u$. As in the proposed model, each user is associated with the BS that provides the best received signal strength. From this, we calculate the actual load being served by each BS in terms of the number of users, which we denote by $N_{x_i}$ for a BS located at $x_i$. Assuming the number of orthogonal resource blocks, e.g., time-frequency resource blocks in LTE~\cite{GhoZhaB2010}, to be $M$, the activity factor of a BS in each resource block can be expressed as $p_{x_i} = N_{x_i}/M$ as discussed in Section~\ref{sec:sysmod}. To keep the setup simple, we consider the regime where the probability of having $N_{x_i} > M$ for any BS is small and whenever it happens, the activity factor for that BS is assumed to be $1$. For this setup, the coverage probability results are presented in Fig.~\ref{fig:SystemSim}.

For a meaningful comparison of this simulation result with the analytical results, we first need to find analytical expressions of the activity factors $p_i$ as a function of the user density $\lambda_u$. For this, we leverage Corollary $2$ of~\cite{DhiGanJ2012}, where it is shown that the fraction of users served by $j^{th}$ tier is given by:
\begin{align}
\bar{N}_j = \frac{\lambda_j \left(P_j/\beta_j  \right)^{\frac{2}{\alpha}}}{\sum_{i=1}^K \lambda_i \left(P_i/\beta_i  \right)^{\frac{2}{\alpha}}}.
\end{align}
Using this result, the average number of users served by a $j^{th}$ tier BS (average load) is $\frac{\lambda_u}{\lambda_j}\bar{N}_j$. Therefore, assuming $M$ resource blocks, the activity factor in a randomly chosen resource block is:
\begin{align}
p_j = \frac{1}{M} \frac{\lambda_ u \bar{N}_j}{\lambda_j} = \frac{\lambda_u}{M} \frac{\left(P_j/\beta_j  \right)^{\frac{2}{\alpha}}}{\sum_{i=1}^K \lambda_i \left(P_i/\beta_i  \right)^{\frac{2}{\alpha}}}.
\end{align}
We use this analytical result for the the load factors in the coverage probability results derived in the paper to plot the analytical results as a function of the user density in Fig.~\ref{fig:SystemSim}. Comparing this result with the numerical result obtained form a detailed simulation, we note that the two are reasonably close, especially relative to the previously known results for the fully-loaded system. This validates the two assumptions mentioned in the starting of this discussion.

\begin{figure}[t!]
\centering
\includegraphics[width=\columnwidth]{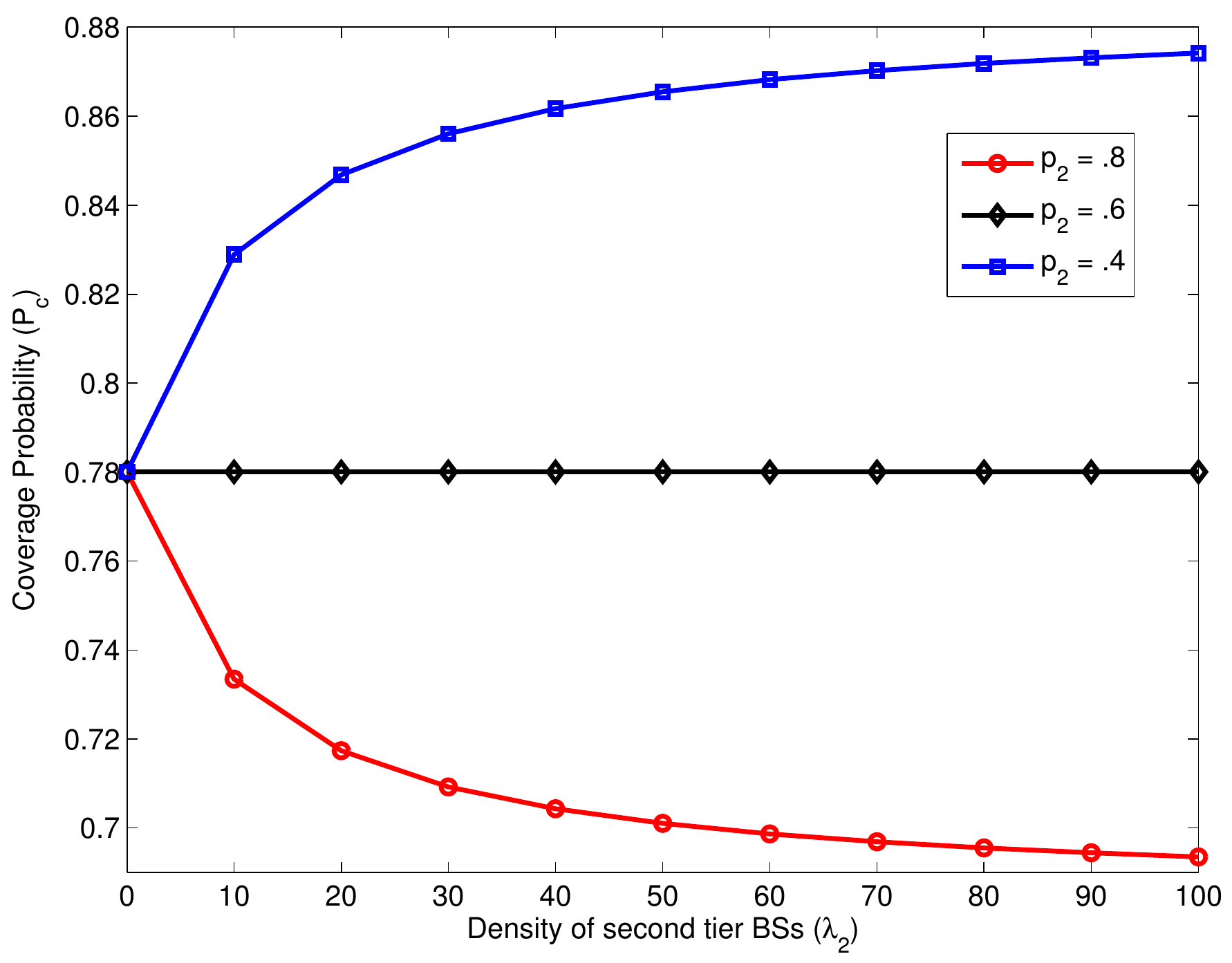}
\caption{Coverage probability in a two tier network as a function of $\lambda_2$ ($\T =[1, 1]$, $P = [1, .01]$, $\lambda_1 = 1$, $p_1 = .6$ and $\alpha = 3.8$).}
\label{fig:2tier_pc_L2}
\end{figure}

\subsection{Scale Invariance and Effect of Adding Small Cells}
We now consider a two tier system and plot the coverage probability as a function of the density of second tier for various BS activity factors in Fig.~\ref{fig:2tier_pc_L2}. The target $\sir$ is fixed to be the same for both the tiers. We first note that the network is invariant to the changes in density when $p_1 = p_2$ as discussed in the last section. More importantly, we note that the coverage probability increases with $\lambda_2$ when the second tier BSs are less active than the first tier. This is an important result from the perspective of small cells, which are generally less active than macrocell BSs. Therefore, the coverage probability of the network should increase with the addition of small cells in this regime. This is a strong rebuttal to the viewpoint that unplanned infrastructure might bring down a cellular network due to increased interference. On the other hand, if a tier of BSs is added which is more active than the macrocells, the coverage would decrease, although this case seems pretty unlikely given the high load handled by the macrocells.

\begin{figure}[t!]
\centering
\includegraphics[width=\columnwidth]{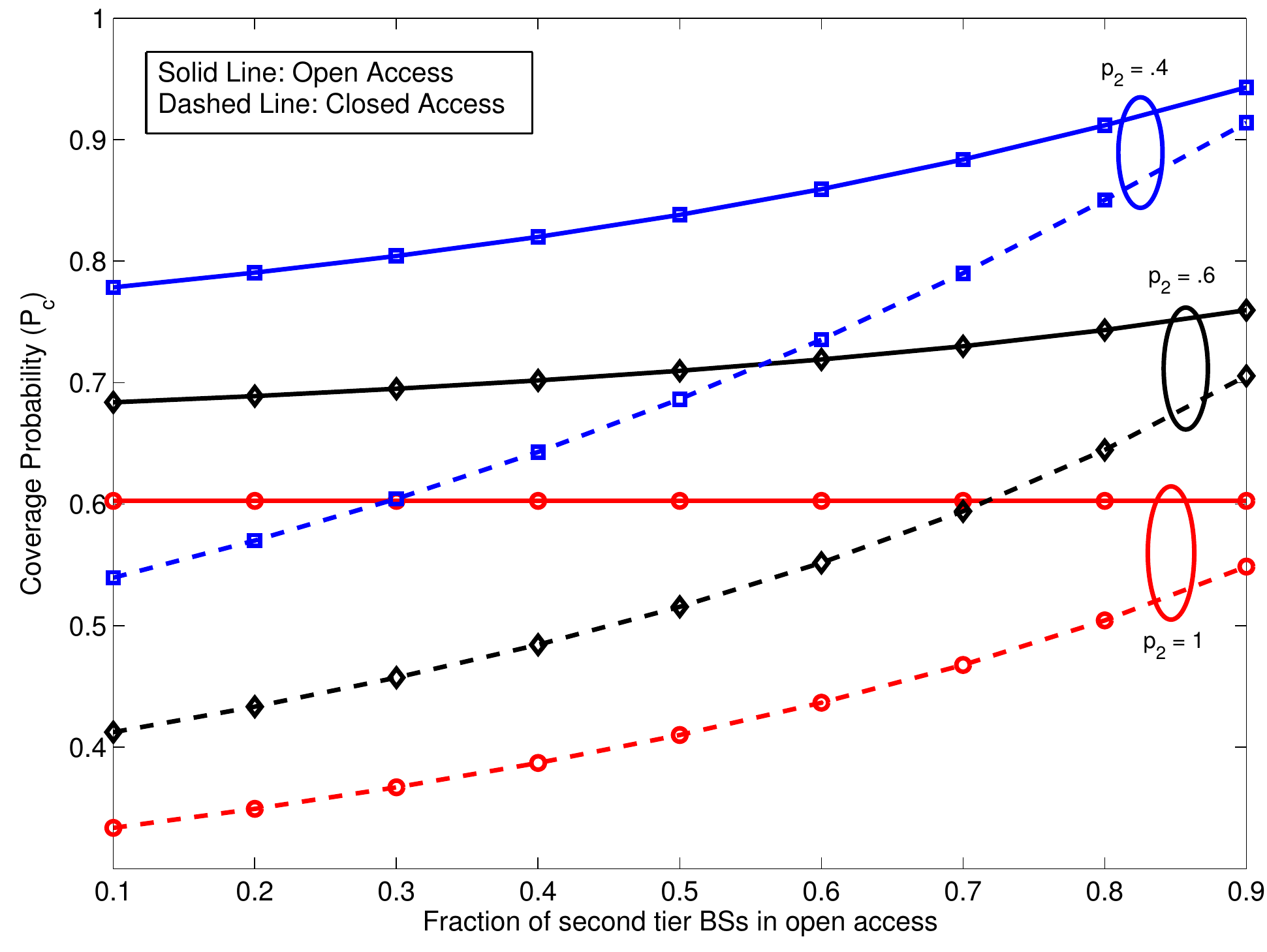}
\caption{Coverage probability in a two tier network as a function of the fraction of the second tier BSs in open access ($\T =[1, 1]$, $P = [1, .01]$, $\lambda_2^{(c)} = 10 \lambda_1$, $p_1 = 1$ and $\alpha = 3.8$). The density of second tier BSs in open access $\lambda_2^{(o)} = \frac{f}{1-f} \lambda_2^{(c)}$, where $f$ is the fraction of BSs in open access.}
\label{fig:OAvsCA}
\end{figure}

\subsection{Open vs Closed Access} So far in this section we have only studied open access networks, where a mobile user can access any BS in the network. We now study the effect of closed access on the coverage probability, with a focus on a particular scenario of interest where a certain fraction of BSs of a particular tier are in closed access while the others are in open access. This scenario is especially important in the current HCNs, where the access permissions of a particular small cell might be different for different set of users. It is easy to argue that this scenario can be visualized as a special case of the general model developed in this paper. To understand this, assume that a fraction $f$ of BSs of a certain tier are in open access -- we assume that a BS is in open or closed access independent of the other BSs. Therefore, the density of BSs in open access $\lambda_i^{(o)}$ and closed access $\lambda_i^{(c)}$ can be evaluated from the following two equations:
\begin{align}
f &= \frac{\lambda_i^{(o)}}{\lambda_i^{(o)} + \lambda_i^{(c)}}\\
\lambda_i &= \lambda_i^{(o)} + \lambda_i^{(c)}.
\end{align}
Now this tier can be divided into two tiers, one with density $\lambda_i^{(o)}$, which is in open access, and other with density $\lambda_i^{(c)}$, which is in closed access -- both form independent PPPs.

To study this scenario in detail, we consider a two tier HCN, where the first tier is in open access and fraction $1-f$ of BSs of the second tier is in closed access. For this scenario, the coverage probability as a function of $f$ is presented in Fig.~\ref{fig:OAvsCA} for various load scenarios. The results confirm the intuition that the gap in open and closed access results reduces when the value of $f$ is increased. More interestingly, the gap is smaller when the second tier BSs are lightly loaded. This implies that the effect of interference due to closed access small cells on coverage probability is negligible if there are enough small cells in open access.

\section{Conclusions}
In this paper, we have incorporated a flexible notion of BS load in random spatial models for $K$-tier HCNs by introducing a new idea of {\em conditionally thinning} the interference field, conditional on the connection of a typical mobile to its serving BS. We have shown that this conditional thinning is a natural way of modeling different levels of load on different types of BSs arising mainly from the differences in their coverage footprints. We observe that the fully loaded models are extremely pessimistic in terms of coverage, and the analysis shows that adding lightly loaded access points (e.g. pico or femtocells) to the macrocell network always increases coverage probability.

This work has numerous extensions. Firstly, it is important to exactly model the temporal and spatial correlation in the BS activity factors and develop tools to incorporate it in the framework developed in this paper. Secondly, it is important to understand the coverage and rate trends for this conditional-thinning using more realistic spatial models that include inter-point interactions, such as Gibbs models~\cite{TayDhiC2012}. Another related idea is to define BS load from queuing perspective and then study spatial and temporal dynamics simultaneously. Since this work focused only on the downlink, an interesting extension is to develop a similar framework for uplink possibly using some of the ideas recently developed in~\cite{NovDhiJ2012}.

\appendices
\section{\label{app:Proof-of-Lemma-main}Proof of Lemma \ref{lem:one}:}
Being consistent with the definition of $\calT$, we note that:
\begin{equation} \calT^m=\nb1 \left(\max_{i\in \calK}\frac{M(\Psi_i)}{\delta_i}<I\right) I^{-2m/\alpha}.\end{equation}
To proceed with the proof, we represent $I^{-2m/\alpha}$ in terms of $\Gamma(x)$ as:
\begin{equation}I^{-2m/\alpha}= \frac{1}{\Gamma(2m/\alpha)}\int_0^\infty e^{-sI} s^{-1+\frac{2m}{\alpha}}\nrmd s, \ m\geq 1,\end{equation}
where $\Gamma(x)$ is the standard gamma function. Using this representation of $I^{-2m/\alpha}$ we can express $\E[\ncalT^m]$ as:
\begin{equation}
\E\left[\nb1 \left(\max_{i\in \calK}\frac{M(\Psi_i)}{\delta_i}<I\right)  \frac{1}{\Gamma\left(\frac{2m}{\alpha}\right)}\int_0^\infty e^{-sI} s^{-1+\frac{2m}{\alpha}}\nrmd s\right]. \end{equation}
Using Fubini's theorem, we can exchange the expectation and the inner integral to obtain
\begin{equation}
\frac{1}{\Gamma\left(\frac{2m}{\alpha}\right)}\int_0^\infty s^{-1+\frac{2m}{\alpha}}\E\left[ e^{-sI} \nb1 \left(\max_{i\in \calK}\frac{M(\Psi_i)}{\delta_i}<I\right)\right] \nrmd s.
\label{eq:ET1}
\end{equation}
Under the assumption $\beta_i > 1,\ \forall\ i$, we know that only one BS in the whole network can establish a downlink connection with a typical mobile. Hence,
\begin{equation}
\nb1 \left(\max_{i\in \calK}\frac{M(\Psi_i)}{\delta_i}>I\right) = \sum\limits_{i=1}^K \sum\limits_{x\in \Psi_i}\nb1\left(\sir(x) > \T_i \right),
\end{equation}
where $\sir(x)$ is the received $\sir$ when a typical mobile is camped to the BS located at $x\in \Psi_i$. Using this expression, the expectation term of \eqref{eq:ET1} can be written as:
\begin{align}
&\E \left[ e^{-sI} \nb1 \left(\max_i\frac{M(\Psi_i)}{\delta_i}<I\right) \right] \nonumber \\
&=\E \left[e^{-sI} \right]-\sum_{i=1}^K \E \left[ e^{-sI}\sum_{x\in \Psi_i}\nb1 (\sir(x)>\beta_i) \right].
\label{eq:ET2}
\end{align}
From Lemma~\ref{Lemma2}, we know the Laplace transform of total interference and hence the first term in the above expression can be directly written as:
\begin{equation} \E\left[ e^{-sI} \right]= \exp\left(-s^{2/\alpha} C(\alpha)\sum_{l=1}^K p_l\lambda_l P_l^{2/\alpha}\right). \label{eq:ET2_term1} \end{equation}
To evaluate the expectation in the second term of \eqref{eq:ET2}, we first denote the effective interference as $I' = I - P_i h_x \|x\|^{-\alpha}$ and note that the Laplace transforms of $I$ and $I'$ are the same. The expectation can now be simplified as:
\begin{align}
&\E \left[e^{-sI}\sum_{x\in \Psi_i}\nb1 (\sir(x)>\beta_i) \right] \nonumber \\
&= \E \left[\sum_{x\in \Psi_i} \exp(-sI' + P_i h_x \|x\|^{-\alpha}) \nb1 \left(\frac{P_i h_x \|x\|^{-\alpha}}{I'}>\beta_i \right) \right] \\
& \stackrel{(a)}{=} \E \left[\sum_{x\in \Psi_i}  e^{-sI'} \E_{h_x}\left[e^{-P_i h_x \|x\|^{-\alpha}}\nb1\left( h_x > \T_i I' P_i^{-1} \|x\|^{\alpha} \right) \right]   \right]\\
&\stackrel{(b)}{=} \E \left[\sum_{x\in \Psi_i} \frac{\E_{I'} \left[\exp(-I'(s(1+\beta_i)+\beta_iP_i^{-1}\|x\|^{\alpha}))\right]}{1+sP_i\|x\|^{-\alpha}}   \right],
\end{align}
where $(a)$ follows from the fact that fading is independent of all the other random variables and $(b)$ follows from the fact that $h_x \sim \exp(1)$. Now, using the Laplace transform of $I'$ and recalling $\eta = C(\alpha) \sum_{l=1}^K \lambda_l p_l P_l^{2/\alpha}$, it can be further simplified to:
\begin{equation}
\E \left[\sum_{x\in \Psi_i} \frac{\exp(-\eta (s(1+\beta_i)+\beta_iP_i^{-1}\|x\|^{\alpha})^\frac{2}{\alpha})}{1+sP_i\|x\|^{-\alpha}}  \right],
\end{equation}
and using Campbell Mecke theorem~\cite{StoKenB1995} to:
\begin{equation}
\lambda_ip_i\int_{\R^2} \frac{\exp(-\eta(s(1+\beta_i)+\beta_iP_i^{-1}\|x\|^{\alpha})^\frac{2}{\alpha})}{1+sP_i\|x\|^{-\alpha}} \nrmd x.
\label{eq:ET2_term2}
\end{equation}
With this we have now simplified both the terms of \eqref{eq:ET2} given respectively by \eqref{eq:ET2_term1} and \eqref{eq:ET2_term2}. We now substitute the first term in \eqref{eq:ET1} and evaluate the integral with respect to $s$ as:
\begin{align}
\int_0^\infty s^{-1+2m/\alpha}\exp\left(-\eta s^{2/\alpha} \right)\nrmd s
&= \frac{\eta^{-m}\alpha(m-1)!}{2},
\end{align}
where the solution follows from the substitution $s^{2/\alpha} \rightarrow y$ followed by integration by parts. Now substituting the second term (given by \eqref{eq:ET2_term2}) in \eqref{eq:ET1}, we get the following integral:
\begin{equation}
\lambda_ip_i\int_0^\infty\int_{\R^2}\frac{s^{-1+2m/\alpha}e^{-\eta(s(1+\beta_i)+\beta_iP_i^{-1}\|x\|^{\alpha})^\frac{2}{\alpha}}}{1+sP_i\|x\|^{-\alpha}}\nrmd x \nrmd s .
\end{equation}
Now use the substitution $(sP_i)^{-1/\alpha}x \to x$, which leads to
\begin{equation}
\lambda_ip_i\int_0^\infty\int_{\R^2}\frac{s^{-1+2m/\alpha}e^{-\eta s^\frac{2}{\alpha}((1+\beta_i)+\beta_i\|x\|^{\alpha})^\frac{2}{\alpha}}}{1+ \|x\|^{-\alpha}} (s P_i)^\frac{2}{\alpha}\nrmd x \nrmd s .
\end{equation}
Now exchange the integrals to obtain
\begin{equation}
\lambda_ip_iP_i^\frac{2}{\alpha}\int_{\R^2}\int_0^\infty\frac{s^{-1+\frac{2(m+1)}{\alpha}}e^{-\eta s^\frac{2}{\alpha}((1+\beta_i)+\beta_i\|x\|^{\alpha})^\frac{2}{\alpha}}}{1+ \|x\|^{-\alpha}} \nrmd s \nrmd x .
\end{equation}
Now the inner integral (with respect to $s$) can be evaluated directly using the definition of $\Gamma(x)$ function or using the substitution $s^{2/\alpha}\to s$ to obtain the below integral.
\begin{align}
\frac{\lambda_ip_iP_i^\frac{2}{\alpha}\alpha m!}{2\eta^{m-1}} \int_{\R^2}\frac{\nrmd x}{(1+\|x\|^{-\alpha})(1+\beta_i+\beta_i\|x\|^\alpha)^{\frac{2}{\alpha}(m+1)}}.
\end{align}
Now the above integral can be expressed as:
\begin{equation}
\frac{1}{(1+\T_i)^{\frac{2}{\alpha}(m+1)}} \int_{\R^2}\frac{\nrmd x}{(1+\|x\|^{-\alpha})(1+\frac{\T_i}{1+\T_i}\|x\|^\alpha)^{\frac{2}{\alpha}(m+1)}}
\end{equation}
Now using the substitution $1+\frac{\beta_i}{1+\beta_i}\|x\|^\alpha \to t^{-1}$, the above expression can be simplified to
\begin{align}
&\frac{2\pi\beta_i^{-2/\alpha}}{\alpha(1+\beta_i)^{2m/\alpha}}\frac{\Gamma(2m/\alpha)\Gamma(1+2/\alpha)}{\Gamma(1+(m+1)2/\alpha)}\nonumber \\
&{}_2F_1(1,2m/\alpha,1+(m+1)2/\alpha,(1+\beta_i)^{-1}),
\end{align}
where ${}_2F_1$ is the generalized hypergeometric function. Combining all the above we obtain the result. \hfill \IEEEQED

\begin{biography}
[{\includegraphics[width=1in,height=1.25in,clip,keepaspectratio]{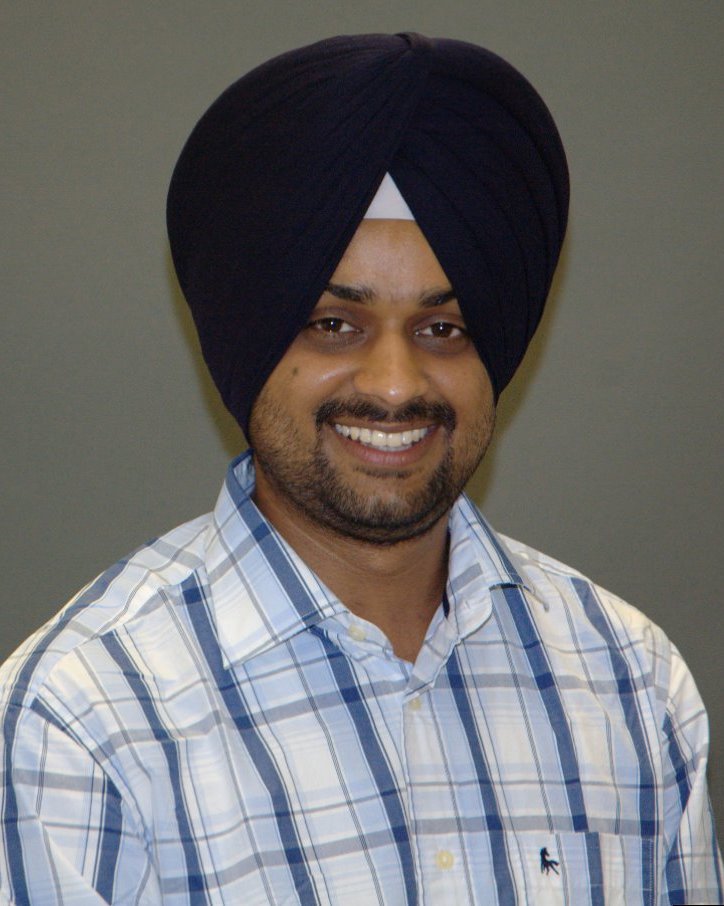}}]
{Harpreet S. Dhillon}
(S'11) received the B.Tech. degree in Electronics and Communication Engineering from IIT Guwahati, India, in 2008 and the M.S. in Electrical Engineering from Virginia Tech in 2010. He is currently a Ph.D. student at The University of Texas at Austin, where his research has focused on the modeling and analysis of heterogeneous cellular networks using tools from stochastic geometry, point process theory and spatial statistics. His other research interests include interference channels, multiuser MIMO systems and cognitive radio networks. He is the recipient of the Microelectronics and Computer Development (MCD) fellowship from UT Austin and was also awarded the Agilent Engineering and Technology Award 2008. He has held summer internships at Alcatel-Lucent Bell Labs in Crawford Hill, NJ, Samsung Dallas Technology Lab in Richardson, TX, Qualcomm Inc. in San Diego, CA, and Cercom, Politecnico di Torino in Italy.
\end{biography}

\begin{biography}
[{\includegraphics[width=1in,height=1.25in,clip,keepaspectratio]{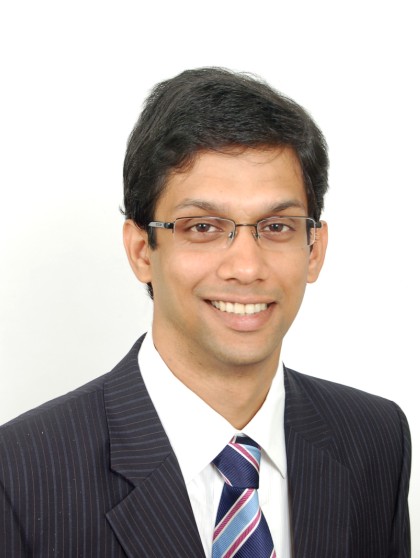}}]
{Radha Krishna Ganti}
(S'01, M'10) is an Assistant Professor at the Indian Institute of Technology Madras, Chennai, India. He was a Postdoctoral researcher in the Wireless Networking and Communications Group at UT Austin from 2009-11. He received his B. Tech. and M. Tech. in EE from the Indian Institute of Technology, Madras, and a Masters in Applied Mathematics and a Ph.D. in EE form the University of Notre Dame in 2009. His doctoral work focused on the spatial analysis of interference networks using tools from stochastic geometry. He is a co-author of the monograph {\em Interference in Large Wireless Networks} (NOW Publishers, 2008).
\end{biography}

\begin{biography}
[{\includegraphics[width=1in,height=1.25in,clip,keepaspectratio]{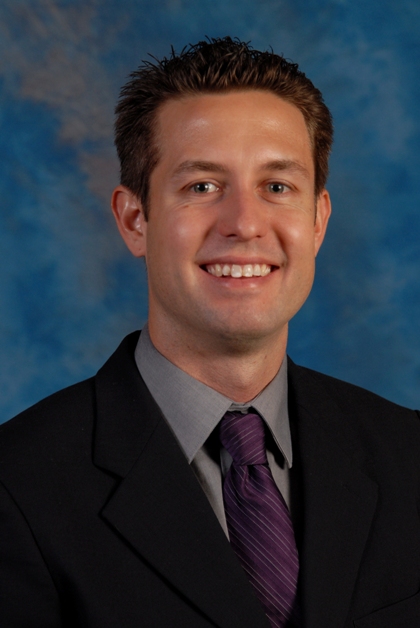}}]
{Jeffrey G. Andrews} (S'98, M'02, SM'06, F'13) received the B.S. in Engineering with High
Distinction from Harvey Mudd College in 1995, and the M.S. and Ph.D. in Electrical
Engineering from Stanford University in 1999 and 2002, respectively. He is a Professor in the Department of Electrical and Computer Engineering at the University of Texas at Austin, where he was the Director of the Wireless Networking and Communications Group (WNCG) from 2008-12. He developed Code Division Multiple Access systems at Qualcomm from 1995-97, and has consulted for entities including Verizon, the WiMAX Forum, Intel, Microsoft, Apple, Samsung, Clearwire, Sprint, and NASA. He is a member of the Technical Advisory Boards of
Accelera and Fastback Networks.

Dr. Andrews is co-author of two books, \emph{Fundamentals of WiMAX} (Prentice-Hall, 2007) and \emph{Fundamentals of LTE} (Prentice-Hall, 2010), and holds the Earl and Margaret Brasfield Endowed Fellowship in Engineering at UT Austin, where he received the ECE department's first annual High Gain award for excellence in research. He is a Senior Member of the IEEE, a Distinguished Lecturer for the IEEE Vehicular Technology Society, served as an associate editor for the \textsc{IEEE Transactions on Wireless Communications} from 2004-08, was the Chair of the 2010 IEEE Communication Theory Workshop, and is the Technical Program co-Chair of ICC 2012 (Comm. Theory Symposium) and Globecom 2014.  He is an elected member of the Board of Governors of the IEEE Information Theory Society and an IEEE Fellow.

Dr. Andrews received the National Science Foundation CAREER award in 2007 and has been co-author of five best paper award recipients, two at Globecom (2006 and 2009), Asilomar (2008), the 2010 IEEE Communications Society Best Tutorial Paper Award, and the 2011 Communications Society Heinrich Hertz Prize.  His research interests are in communication theory, information theory, and stochastic geometry applied to wireless cellular and ad hoc networks.
\end{biography}  \vfill

\end{document}